\documentclass[12pt]{amsart}
\usepackage{amssymb,mathrsfs,bbold}
 \usepackage[all]{xy}
\usepackage{color}
\theoremstyle{plain}
\newtheorem{theorem}{Theorem}[section]
\newtheorem{proposition}[theorem]{Proposition}
\newtheorem{corollary}[theorem]{Corollary}
\newtheorem{lemma}[theorem]{Lemma}

\theoremstyle{definition}
\newtheorem{remark}[theorem]{Remark}

\newcommand{\abs}[1]{\lvert#1\rvert}
\newcommand{\norm}[1]{\lVert#1\rVert}
\newcommand{\bigabs}[1]{\bigl\lvert#1\bigr\rvert}

\newcommand{\term}[1]{{\textit{\textbf{#1}}}}

\DeclareMathOperator{\Span}{Span}

\def\one{\mathbb 1}

\title{Option spanning beyond $L_p$-models}

\date{\today}

\keywords{Spanning of options, market completeness, arbitrage, Kreps-Yan
Theorem, order continuous dual}
\subjclass[2010]{Primary: 91G20, 91B25, Secondary: 46B42}

\author{N.~Gao}
\address{School of Mathematics, Southwest Jiaotong University, Chengdu, Sichuan,
610000, China.}
\email{ngao@home.swjtu.edu.cn}

\author{F.~Xanthos}
\address{Department of Mathematics, Ryerson University, 350 Victoria St.,
Toronto, ON, M5B 2K3, Canada.}
\email{foivos@ryerson.ca}

\thanks{The authors are partially supported by an NSERC grant.}

\begin{document}
\begin{abstract}
The aim of this paper is to study the spanning power of options in a static
financial market that allows non-integrable assets. Our findings extend and
unify the results in \cite{GALVANI:09,GALVANI:10,Nachman} for $L_p$-models. We
also apply the spanning power properties to the pricing problem. In particular,
we show that prices on call and put options of a limited liability asset can be
uniquely extended by arbitrage to all marketed contingent claims written on the
asset.
\end{abstract}

\maketitle

\section{Introduction}

Throughout this paper, $\Omega$ stands for the state space of a financial
market, $\Sigma$ stands for the $\sigma$-algebra modelling the market
information structure, and $\mathbb{P}$ stands for a probability over
$(\Omega,\Sigma)$. The space of contingent claims, $X$, is modelled as an ideal
(i.e., solid subspace) of $L_0(\Sigma)$
containing the constant functions, which represent investments in the riskless
asset. A claim displays limited liabilities if it is positive.
For a limited liability claim $f$, its option space is the collection of all
portfolios of call and put options written on $f$, which can be identified as
follows:$$O_f=\Span\big\{\one, (f-k)^+:k\in\mathbb{R}\big\}.$$
The space of all contingent claims written on $f$
is identified as the space of all functions measurable with respect to
$\sigma(f)$, the sub-$\sigma$-algebra generated by $f$, i.e.,
$$L_0(\sigma(f)).$$

A stream of research has been devoted to the study of spanning power of options
on $f$, i.e., the size of $O_f$. In the seminal paper \cite{ROSS}, Ross showed
that if the state space $\Omega$ is finite
then the options on $f$ span the space of contingent claims written on $f$,
i.e., $$O_{f}=L_0(\sigma(f)),$$
and if, in addition, $f$ is one-to-one, then the option space of $f$ completes
the market, i.e., $$O_f=L_0(\Sigma).$$
These elegant results of Ross have inspired many successive contributions to
the study of options. See e.g.~\cite{BAPT05,FX1,FX2} for related results on
finite state spaces. In particular, they have also been examined for financial
markets with infinite state spaces

Nachman proved in \cite{Nachman} that if $X=L_p(\Sigma)$ ($1\leq p<\infty$),
then the options on $f$ span the space of contingent claims written on the asset
in two ways: approximating by a.e.~convergence or by $p$-th mean convergence.
Precisely, it was proved that an asset $x\in L_p(\Sigma)$ is a contingent claim
on $f\in L_p(\Sigma)$ iff there exists a sequence of portfolios of options on
$f$ converging a.e.~to $x$ iff there exists a sequence of portfolios of options
on $f$ converging in the $p$-th mean to $x$. That is,
$$\overline{O_f}^{a.e.}\cap
L_p(\Sigma)=\overline{O_f}^{\norm{\cdot}_p}=L_0(\sigma(f))\cap
L_p(\Sigma).\eqno(1)$$
Galvani (\cite{GALVANI:09}) and Galvani and Troitsky (\cite{GALVANI:10}) proved
further that if $\Omega$ is a Polish space equipped with the Borel
$\sigma$-algebra and $f$ is one-to-one and bounded, then $O_f$ completes the
market $X=L_p(\Sigma)$ ($1\leq p\leq \infty$). That is, for $1\leq p<\infty$,
$$\overline{O_f}^{a.e.}\cap L_p(\Sigma)=\overline{O_f}^{\norm{\cdot}_p}=
L_p(\Sigma),\eqno(2)$$
and$$\overline{O_f}^{a.e.}\cap
L_\infty(\Sigma)=\overline{O_f}^{w^*}=L_\infty(\Sigma).\eqno(3)$$

In this paper, we explore the spanning power of options in general spaces of
contingent claims. Our contributions here are two-fold. Firstly, the spaces of
contingent claims in our
setting can be modelled as any ideal of $L_0(\Sigma)$ which contains the
constant functions and admits a strictly positive order
continuous linear functional. This framework includes not only the $L_p$-space
($1\leq
p\leq \infty$) models, but also the much wider class of Orlicz space models as
well as many non-integrable space models which have been extensively used in the
theory of risk measures (see e.g.
\cite{Biagini:08,Biagini:10,Cheridito:09,Delbaen:09,GAOX:14,Orihuela:12}).

Secondly, we provide an approach to unify the norm and $w^*$-topologies used in
the results of Nachman, Galvani and Troitsky, and thus give more comprehensive
insight into the general structures of option spaces; see Theorem~\ref{main} and
Remark~\ref{rem}. The unification in our approach is due to the use of the
topology $\sigma(X,X_n^\sim)$, where $X_n^\sim$ is the set of all order
continuous linear functionals on $X$.

Observe that $(X_n^\sim)_+$ has a natural connection with linear
pricing functionals. Recall that a linear pricing functional $\phi$ on $X$ is
given by a state-price density
$y\geq 0$ via integration, i.e.$$\phi(x)=\int_\Omega xy\,\mathrm{d}\mathbb{P}
\text{
for all }x\in X,$$
where $y$ is measurable and satisfies $\int_\Omega
\abs{xy}\mathrm{d}\mathbb{P}<\infty$ for all $x\in X$. By Dominated Convergence
Theorem, it is easily seen that $\phi$ is order continuous 
on $X$. Conversely, by Radon-Nikodym theorem, one can easily see that each
positive order
continuous linear functional on $X$ has a positive density, and thus is a linear
pricing functional.
Therefore, $(X_n^\sim)_+$ is just the collection of linear pricing functionals
on $X$.

Because of this link, we are able to apply Theorem~\ref{main} and shed light on
the following general problem, raised in \cite{BROWN:91}: ``Under what
circumstances can prices on the marketed assets or basic derivative assets be
uniquely extended by arbitrage to prices on all derivative assets in a large
class and when is such an extension unique?''
In Theorem~\ref{arbitrage}, we prove that when the arbitrage condition is
understood as a no-free lunch condition (see \cite{Kreps81}), one can extend
uniquely the prices on $O_f$ to the marketed contingent claims written on $f$.

Finally, we mention that there is a stream of works studying market
completion using options in a continuous time setting. In this framework, the
model
is said to be complete, if any contingent claim payoff can be obtained as the
terminal value of a self-financing trading strategy. 
We refer the reader to the recent papers
\cite{DAVIS02,HUGO12,RH13,ROMA97,SCH15}.

\section{Preliminary results}

We refer to \cite{ALIP:06, ALIP:07} for all unexplained terminology and standard
facts on vector and Banach lattices.
A vector subspace $Y$ of a vector lattice $X$ is called a \term{sublattice} if
$\abs{y}\in Y$ whenever $y\in Y$; in this case, $y_1\wedge y_2$ and $y_1\vee
y_2$ both belong to $ Y$ whenever $y_1,y_2\in Y$. A subspace $Y$ is called an
\term{ideal} (or a \term{solid subspace}) of $X$, if $\abs{x}\leq \abs{y}$ and
$y\in Y$ imply $x\in Y$.
A linear functional $\phi$ on a vector lattice $X$ is said to be \term{order
continuous} if $\phi(x_\alpha)\rightarrow 0$ whenever $x_\alpha\xrightarrow{o}0$
in $X$. The collection of all order continuous linear functionals on $X$ is
denoted by $X_n^\sim$ and is called the \term{order continuous dual} of $X$. A
linear functional $\phi$ on $X$ is said to be \term{positive} if $\phi(x)\geq 0$
whenever $x\geq 0$, and is said to be \term{strictly positive} if $\phi(x)>0$
whenever $x>0$.

The following lemma will be used. Recall first that a vector lattice is said to
be \term{order complete} (or \term{Dedekind complete}) if every order bounded
above subset has a supremum, and is said to have the \term{countable sup
property} if any subset having a supremum possesses a countable subset with the
same supremum. A subset $A$ of a vector lattice $X$ is said to be \term{order
closed} if $x\in A$ whenever there exists a net $(a_\alpha)$ in $A$ such that
$a_\alpha\xrightarrow{o}x$ in $X$.

\begin{lemma}\label{o-closed}
Let $X$ be an order complete vector lattice with the countable sup property and
$Y$ be a sublattice of $X$. Then $Y$ is order closed in $X $ iff for any
increasing sequence in $Y$ which is order bounded above in $X$, its supremum in
$X$ also lies in $Y$.
\end{lemma}

Given a probability space $(\Omega,\Sigma,\mathbb{P})$, denote by $L_0(\Sigma)$
the space of all real-valued measurable functions (\emph{modulo
a.e.~equality}). 
We use $\one$ to denote the constant one function.
Recall that $L_0(\Sigma)$ is a vector lattice, endowed with the natural order:
$f\leq g$ iff $f(\omega)\leq g(\omega)$ for a.e.~$\omega\in \Omega$.
By \cite[Lemma~2.6.1]{Meyer-Nieberg:91}, it is easily seen that any ideal of
$L_0(\Sigma)$ is order complete and has the countable sup property. Hence,
Lemma~\ref{o-closed} is applicable to them. Recall also that
$f_n\xrightarrow{o}0$ in an ideal $X$ of $L_0(\Sigma)$ if and only if
$f_n\xrightarrow{a.e.}0$ and $(f_n)_{n=1}^\infty$ is \emph{order bounded} in
$X$, i.e., there exists $f\in X$ such that $\abs{f_n}\leq f$ a.e.~for each
$n\geq 1$.
We remark that the class of ideals of $L_0(\Sigma)$ which admit strictly
positive order continuous linear functionals is very large. For example, by
\cite[Proposition~5.19]{GAOTX:15}, all Banach function spaces (i.e., ideals of
$L_0(\Sigma)$ endowed with complete lattice norm), including all Orlicz spaces,
are as
such.

For a subset $Y$ of $L_0(\Sigma)$, define $\sigma(Y)$ to be the smallest
sub-$\sigma$-algebra of $\Sigma$ which makes all members in $Y$ measurable and
contains all $\mathbb{P}$-null sets.  
Denote by $L_0(\sigma(Y))$ the set of all functions in $L_0(\Sigma)$ which are
measurable with respect to $\sigma(Y)$.
Clearly, $Y\subset L_0(\sigma(Y))$.
If $Y=\{f\}$, we write $\sigma(f)$ instead of $\sigma(\{f\})$, for the sake of
simplicity.
The following result is
an improved and generalized market completeness theorem in the sense of Green
and Jarrow (\cite[Theorem~1]{GREEN:87}).

\begin{lemma}\label{GJ-comp}
Let $X$ be an ideal of $L_0(\Sigma)$ and $Y$ be a sublattice of $X$ such that
$\mathbb{1} \in Y$. Then the following are equivalent:
\begin{enumerate}
\item[(a)] $Y$ is order closed in $X$,
\item[(b)] $Y=L_0(\sigma(Y)) \cap X$.
\end{enumerate}
\end{lemma}

\section{Main results}

In this section, the space of contingent claims, $X$, is always modelled as an
ideal of
$L_0(\Sigma)$ over a given probability space $(\Omega,\Sigma,\mathbb{P})$ that
contains the constant functions and admits a strictly positive order
continuous linear functional. 
Our first main result is as follows.

\begin{theorem}\label{main}
Let $f$ be a limited liability claim in $X$. For
a claim $g\in X$, the following are equivalent:
\begin{itemize}
\item[(a)] $g$ is a contingent claim written on $f$, i.e., $g\in 
L_0(\sigma(f)) \cap X$,
\item[(b)] $g$ can be approximated by portfolios of options on $f$ in the
${\sigma(X,X_n^\sim)}$-topology, i.e., $g\in
\overline{O_{f}}^{{\sigma(X,X_n^\sim)}}$,
\item[(c)] There exists a sequence $(g_n)$ in $O_{f}$ such that $g_n
\xrightarrow{{a.e.}} g$.
\end{itemize}
\end{theorem}

The following corollary is immediate.

\begin{corollary}\label{main2}
Let $f$ be a limited liability claim in $X$ such that $\sigma(f)=\Sigma$.
Then we have the following:
\begin{itemize}
\item[(a)] The option space of $f$ completes the market in the
$\sigma(X,X_n^\sim)$-topology,
i.e., $\overline{O_{f}}^{{\sigma(X,X_n^\sim)}}=X$,
\item[(b)] The option space of $f$ completes the market by approximating via
a.e.~convergence, i.e., for any $g\in X$, there exists a sequence $(g_n)$ in
$O_{f}$ such that $g_n \xrightarrow{{a.e.}} g$.
\end{itemize}
\end{corollary}

\begin{remark}\label{rem}
Note that our Theorem~\ref{main} and Corollary~\ref{main2} imply both the
aforementioned results of Nachman,
Galvani and Troitsky. Indeed, recall first that, for $1\leq p<\infty$,
$X=L_p(\Sigma)$ is \term{order continuous}, i.e., $x_\alpha\downarrow0$ in $X$
implies $\norm{x_\alpha}\downarrow 0$. In this case, one has $X_n^\sim=X^*$, so
that $\sigma(X,X_n^\sim)$ is just the weak topology on $X$. Thus, by Mazur's
theorem,
$\overline{C}^{\sigma(X,X_n^\sim)}=\overline{C}^w=\overline{C}^{\norm{\cdot}}$
for any convex subset $C$ of $X$. Consequently, it follows that
$$\overline{O_{f}}^{{\sigma(X,X_n^\sim)}}=\overline{O_{f}}^{\norm{\cdot}_p}
.$$
Now it is clear that Equation (1) follows from Theorem~\ref{main}. If, in
addition, $\Omega $ is a Polish space with $\Sigma$ being the Borel algebra, and
$f$ is one-to-one, then it is easily seen that $\sigma(f)=\Sigma$ by
\cite[Theorem~12.29]{ALIP:07}.
Thus, Equation (2) follows from Corollary~\ref{main2}. Equation (3) also follows
from Corollary~\ref{main2}, since $L_\infty(\Sigma)_n^\sim=L_1(\Sigma)$ and thus
$\sigma(L_\infty(\Sigma),L_\infty(\Sigma)_n^\sim)$ is just the $w^*$-topology.
\end{remark}

We now turn to discuss the pricing problem. Our notation and terminology are in
accordance with \cite{Kreps81,Schach02}.

Let $f$ be a fixed asset in $X$ and $\pi$ be a positive linear functional on the
option space $M:=O_{f}$, which is interpreted as a linear pricing functional on
$M$. We denote by $M_0:=\{x \in M \,\, | \,\, \pi(x)=0\}$ the set of all
portfolios of options on $f$ that can be bought or sold with zero price. We
say that $(M,\pi)$ admits  \textbf{no free lunches} (cf.~\cite[Definition
1.3]{Schach02}), if the following holds
$$C\cap X_+=\{0\}, \text{ where } C=\overline{M_0-X_+}^{\sigma(X,X_n^\sim)}.$$

We say that a price, $p$, of an asset $g \in X$ is \textbf{consistent} with
$(M,\pi)$ if there exists  a strictly positive functional $x^*\in X_n^\sim$ such
that $x^*|_M=\pi$ and $x^*(g)=p$ (\cite[Definition, pp.~29]{Kreps81}). The
price of $g \in X$ is said to be \textbf{determined by arbitrage from
$(M,\pi)$} if there is a single price $p$ for $g$ that is consistent with
$(M,\pi)$  (\cite[Definition, pp.~30]{Kreps81}).

\begin{theorem}\label{arbitrage}
Suppose that the space $X$ of contingent claims is a Banach function space in
$L_0(\Sigma)$. Let $f$ be a limited liability asset in $X$ and $\pi$ be a
positive linear functional on the option space $M=O_{f}$. If $(M,\pi)$ admits
no free lunches, then the price of any contingent claim $g \in  L_0(\sigma(f))
\cap X$ is determined by arbitrage from
$(M,\pi)$.
\end{theorem}

The proof of this result essentially depends on the following version of the
Kreps-Yan Theorem, which is of independent interest.

\begin{proposition}\label{Kreps-Yan}
Let $X$ be a Banach function space in $L_0(\Sigma)$. Then the Kreps-Yan theorem
holds true for $\big(X,\sigma(X,X^\sim_n)\big)$. That is, for each
$\sigma(X,X^\sim_n)$-closed cone $C$ in $X$ such that $C\supset -X_+ $ and $C
\cap X_+=\{0\}$, there exists a strictly positive functional $\phi \in X^\sim_n$
such that $\phi|_{C} \leq 0$.
\end{proposition}

The proof of this result (see Section~4) relies on \cite[Theorem~3.1]{Jouini05}.
For more results in this direction, we refer the reader to \cite{ROKH05,ROKH09}.
For no-arbitrage results, we refer the reader to the monograph \cite{Delbaen:06}
and the references therein.

\section{Proofs of Results}

\begin{proof}[Proof of Lemma~\ref{o-closed}]
Let $(y_n)$ be an increasing sequence in $Y$ that is order bounded above in $X$.
Since $X$ is order complete, it follows that $(y_n)$ has a supremum, $x$, in
$X$. Since $(y_n)$ is increasing, it follows that $y_n\uparrow x$ in $X$, so
that
$y_n\xrightarrow{o}x$ in $X$. Thus, if $Y$ is order closed in $X$, then $x\in
Y$. This proves the ``only if'' part.

For the ``if'' part, observe first that, in this case, for any sequence $(y_n)$
in $Y$ which is \emph{order bounded} in $X$, its supremum and infimum in $X$
also lie in $Y$. Indeed, denote by $x$ the supremum of $(y_n)$ in $X$. Put
$z_n=\bigvee_{k=1}^n y_k$. Then $z_n\in Y$ as $Y$ is a sublattice of $X$, and
moreover, the supremum of $(z_n)$ in $X$ is still $x$. Since $(z_n)$ is
increasing, it follows from the ``if'' assumption that $x\in Y$. Replacing
$(y_n)$ with $(-y_n)$, one sees easily that the infimum of $(y_n)$ in $X$ also
lies in $Y$.
Now let $(y_\alpha)\subset Y$ and $x\in X$ be such that
$y_\alpha\xrightarrow{o}x$ in $X$. By passing to a tail, we may assume that
$(y_\alpha)$ is order bounded in $X$. Then since $X$ is order complete, we have
$$\inf_\alpha\sup_{\beta\geq \alpha}\abs{y_\beta-x}=0,$$
where all the suprema and infima are taken in $X$. By the countable sup property
of $X$, we can choose $\{\alpha_n\}_{n=1}^\infty$ such that
$\inf_n\sup_{\beta\geq \alpha_n}\abs{y_\beta-x}=0$. Without loss of generality,
we can assume that $(\alpha_n)$ is increasing. It follows that
$$\inf_n\sup_{m\geq n}\abs{y_{\alpha_m}-x}=0,$$
or equivalently, $y_{\alpha_n}\xrightarrow{o}x$, so that $x=\inf_n\sup_{m\geq
n}y_{\alpha_m}$; cf.~\cite[Theorem 8.16]{ALIP:07}. Applying the preceding
observation to the suprema and then to the infimum, we obtain that $x\in Y$.
\end{proof}

\begin{proof}[Proof of Lemma~\ref{GJ-comp}]
Assume first (b) holds. Let $(f_n)$ be an increasing sequence in $Y$ and $f$
be its supremum in $X$. Then $f_n\uparrow f$ a.e. Since each $f_n$ is
$\sigma(Y)$-measurable, we have that $f$ is also $\sigma(Y)$-measurable, so that
$f\in L_0(\sigma(Y))\cap X=Y$. Thus since $X$ is order complete
and has the countable sup property, Lemma~\ref{o-closed} implies that (a) holds.

Conversely, assume that (a) holds. We first claim that
$\sigma(Y)=\big\{A\in\Sigma: \chi_A\in Y\big\}$\footnote{This is essentially
contained in \cite[Theorem~1]{GREEN:87}.}.
Denote the right hand side by $\mathcal{G}$. We first show that it is a
$\sigma$-algebra. Indeed, it is clear that $\emptyset \in \mathcal{G}$, and that
if $A\in\mathcal{G}$, then $\chi_{A^c}=\one-\chi_A\in Y$, so that
$A^c\in\mathcal{G}$. Now let $(A_k)_1^\infty$ be a sequence of sets in
$\mathcal{G}$. Then
$\chi_{\cup_{k=1}^nA_k}=\bigvee_{k=1}^n\chi_{A_k}\in Y$, and from
$\chi_{\cup_{k=1}^nA_k}\uparrow \chi_{\cup_{k=1}^\infty A_k}$ in $X$, it
follows that $\chi_{\cup_{k=1}^\infty A_k}\in Y$, since $Y$ is order closed.
Therefore, $\bigcup_{k=1}^\infty A_k\in \mathcal{G}$. This concludes the proof
of that $\mathcal{G}$ is a $\sigma$-algebra. Next, we show that each $f\in Y$ is
measurable with respect to $\mathcal{G}$. Indeed, for any real
number $r$, it follows from $Y\ni n(f-r)^+\wedge \one\uparrow
\chi_{\{f>r\}}$ in $X$ that $\chi_{\{f>r\}}\in Y$, so that $\{f>r\}\in
\mathcal{G}$, and $f$ is $\mathcal{G}$-measurable. Clearly, $\mathcal{G}$
contains all $\mathbb{P}$-null sets\footnote{Keep in mind that $\chi_A$ is
identified as $0$ in $L_0(\Sigma)$ if $\mathbb{P}(A)=0$.}. Thus we have
$\sigma(Y)\subset \mathcal{G}$. The reverse inclusion being clear, this
completes the proof of the claim.

It is clear that $Y\subset L_0(\sigma(Y))\cap X$. Now take $f\in
L_0(\sigma(Y))\cap X$. By considering
$f^\pm$, we may assume that $f$ is non-negative. Then we can find a sequence
$(f_n) $ of simple functions which are measurable with respect to $\sigma(Y)$
such that $f_n\uparrow f$ everywhere, so that $f_n\uparrow f$ in $X$. By the
preceding claim, we have that $f_n\in Y$. Therefore, $f\in Y$, and thus
$L_0(\sigma(Y))\cap X\subset Y$. It follows that $Y= L_0(\sigma(Y))\cap X$.
\end{proof}

\begin{proof}[Proof of Theorem~\ref{main}]
We claim that $O_{f}$ is a sublattice of $X$. Indeed, put $Z=\Span\big\{b, s,
(s-k b)^+:k\in\mathbb{R}\big\}$, where $b=f+\one$ and $s=f$. Note that, being an
ideal of $L_0(\Sigma)$, $X$
is order complete, and thus it is uniformly
complete (cf.~\cite[Lemma 1.56]{ALIP:03}). By \cite[Theorem~(1)]{BROWN:91}, it
follows that $Z$ is a sublattice of $X$. Now simply observe that $Z=O_{f}$.
Indeed, the inclusion $O_{f} \subset Z$ is immediate as $f,\one \in Z$ and
$Z$ is closed under lattice operations. For the reverse inclusion, note that
$s=f=(f-0)^+\in O_f$ so that $b\in O_f$ as well. Also, for
$k \geq 1$ we have $(s-kb)^+=0$, and for $k<1$ we have
$(s-kb)^+=(1-k)(f-\frac{k}{1-k}\one)^+ \in O_{f}$.

Assume that (c) holds. Since $Z$ is a sublattice of $X$, by considering the
positive and negative parts, respectively, we may assume that $g\geq 0$ and
$g_n\geq 0$ for all $n$. For any $k\geq 1$, since $g_k\wedge
g_n\xrightarrow{a.e.}g_k\wedge g$ and $(g_k\wedge g_n)_n$ is order bounded in
$X$, it follows that $g_k\wedge g_n\xrightarrow{o}g_k\wedge g$ in $X$, and
therefore, $g_k\wedge g_n\xrightarrow{\sigma(X,X_n^\sim)}g_k\wedge g$ as
$n\rightarrow\infty$.
By the fact that $Z$ is a sublattice again, we have $g_k\wedge g_n\in Z$ for all
$k,n\geq
1$.
Hence, $$g_k\wedge g\in \overline{Z}^{\sigma(X,X_n^\sim)}$$ for any $k\geq 1$.
Now $g_k\wedge g\xrightarrow{a.e.}g$ and $(g_k\wedge g)$ is order bounded in
$X$, we have $g_k\wedge g\xrightarrow{o}g$ in $X$, so that
$g_k\wedge g\xrightarrow{\sigma(X,X_n^\sim)}g$. Therefore, since
$\overline{Z}^{\sigma(X,X_n^\sim)}$ is ${\sigma(X,X_n^\sim)}$-closed, we have
$$g\in \overline{Z}^{\sigma(X,X_n^\sim)}.$$
This proves that (c)$\Rightarrow$(b).

Suppose now (b) holds. Recall that $X_n^\sim $ is a band (i.e., order closed
ideal) of the order dual $X^\sim$ (\cite[Theorem~1.57]{ALIP:06}). It follows
from
\cite[Theorem~3.50]{ALIP:06} that the
dual of $X$ under the topology $\abs{\sigma}(X,X_n^\sim)$ is just $X_n^\sim$.
Therefore, by Mazur's theorem (cf.~\cite[Theorem~3.13]{ALIP:06}), since $Z$ is
convex, $
\overline{Z}^{\sigma(X,X_n^\sim)}=\overline{Z}^{\abs{\sigma}(X,X_n^\sim)}$.
Consequently, there exists a net $(g_\alpha)$ in $Z$ such that
$g_\alpha\xrightarrow{\abs{\sigma}(X,X_n^\sim) }g$. In particular, if
$x^*_0$ is any strictly positive order continuous functional on $X$, then
$$x^*_0(\abs{g_\alpha-g})\rightarrow 0.$$
Take $(\alpha_n)$ such that
$x^*_0\big(\bigabs{g_{\alpha_n}-g}\big)\leq\frac{1}{2^n}$.
Then since $\bigvee_{m=n}^k\bigabs{g_{\alpha_m}-g}\wedge \one\uparrow_k
\sup_{m\geq n}\bigabs{g_{\alpha_m}-g}\wedge \one$, it follows from order
continuity of $x^*_0$ that
\begin{align*}
&x^*_0\big(\sup_{m\geq n}\abs{g_{\alpha_m}-g}\wedge \one\big)=
\lim_k x^*_0\big(\vee_{m=n}^k\bigabs{g_{\alpha_m}-g}\wedge \one\big)\\
\leq &\lim_kx^*_0\Big(\sum_{m=n}^k\bigabs{g_{\alpha_m}-g}\wedge \one\Big)
\leq \frac{1}{2^{n-1}}.
\end{align*}
Therefore, we have $$x^*_0\Big(\inf_{n\geq 1}\sup_{m\geq
n}\bigabs{g_{\alpha_m}-g}\wedge \one\Big)=0,$$
and thus by strict positivity of $x^*_0$, we have
$$\inf_{n\geq 1}\sup_{m\geq n}\bigabs{g_{\alpha_m}-g}\wedge \one=0.$$
If $g_{\alpha_n}\not\xrightarrow{a.e.}g$, then there exist $\varepsilon>0$ and a
measurable set $A$ of positive measure such that
$\limsup_n\abs{g_n(\omega)-g(\omega)}\geq \varepsilon$ for any $\omega\in A$.
Therefore, it is easily seen that $$\inf_{n\geq 1}\sup_{m\geq
n}\bigabs{g_{\alpha_n}-g}\wedge \one\geq (\varepsilon \chi_A)\wedge
\one>0.$$
This contradiction concludes the proof of (b)$\Rightarrow$(c).

Now put $Y=\overline{Z}^{{\sigma(X,X_n^\sim)}}$. Then $Y$ is clearly order
closed in $X$. Moreover, by the preceding paragraph,
$Y=\overline{Z}^{{\abs{\sigma}(X,X_n^\sim)}}$, implying that it is also a
sublattice of $X$ by \cite[Theorem~3.46]{ALIP:06}. Thus by Lemma~\ref{GJ-comp},
$Y=L_0(\sigma(Y))\cap X$.
Since $f\in Y$, it is clear that $\sigma(Y)\supset \sigma(f)$, so that
$$Y=L_0(\sigma(Y))\cap X\supset L_0(\sigma(f))\cap X.$$
For the reverse inclusion, note that, by definition of $O_f$, it is easily seen
that  each $g\in O_{f}$ is measurable with respect to $\sigma(f)$. Now for an
arbitrary $g\in Y$,
we can take, by the implication (b)$\Rightarrow$(c), a sequence $(g_n)$ in
$O_{f}$ such that $g_n\xrightarrow{a.e.}g$. Clearly, $g$ is also
$\sigma(f)$-measurable. Therefore, it follows that $$Y\subset
L_0(\sigma(f))\cap X,$$
and hence $Y=L_0(\sigma(f))\cap X$. This proves (a)$\Leftrightarrow$(b).
\end{proof}

\begin{proof}[Proof of Proposition~\ref{Kreps-Yan}]
We apply \cite[Theorem 3.1]{Jouini05} to $\big(X,\sigma(X,X^\sim_n)\big)$, and
verify that the following Assumptions (C) and (L) are satisfied.\par
Assumption (C): For every sequence $(x_n^*)$ in $X_n^\sim$, there exist strictly
positive numbers $(\alpha_n)$ such that $\sum_{n=1}^\infty \alpha_nx_n^*$
converges in $X_n^\sim$ with respect to the $\sigma(X_n^\sim, X)$-topology.\par
Assumption (L): Any
family $\{x^*_\gamma\}_{\gamma\in\Gamma}$ in $(X^\sim_n)_+$ admits a countable
subfamily $\{x^*_{\gamma_n}\}_{n\geq 1}$ such that, for any $x\in X_+$,
$x^*_{\gamma_n}(x)=0$ for all $n\geq 1$ implies $x^*_{\gamma}(x)=0$ for all
$\gamma\in\Gamma$.

We first verify that Assumption (C) is satisfied. Indeed, since $X$ is a Banach
lattice, we know that the order dual $X^\sim$ equals the norm dual $X^*$
(\cite[Corollary~4.5]{ALIP:06}) and is thus a Banach lattice. By
\cite[Theorem~1.57]{ALIP:06}, $X_n^\sim $ is a band (i.e., order closed
ideal) in $X^\sim=X^*$, and is thus norm closed in $X^*$ by
\cite[Theorem~3.46]{ALIP:06}.
Now for a sequence $(x_n^*) $ in $X_n^\sim$, put
$\alpha_n=\frac{1}{2^n\norm{x_n^*}+1}$ for each $n\geq 1$. Then $\alpha_n$'s are
strictly positive, and $\sum_1^\infty \alpha_nx_n^*$ converges in norm to some
$x^*\in X^*$. Since $X_n^\sim$ is norm closed in $X^*$, it follows that $x^*\in
X_n^\sim$. Clearly, $\sum_1^\infty \alpha_nx_n^*$ also converges to $x^*$ in the
$\sigma(X_n^\sim,X)$-topology.

We now verify that Assumption (L) is also satisfied. 
For the given family $\{x^*_\gamma\}_{\gamma\in\Gamma}$ in $(X_n^\sim)_+$,
put $N_\gamma:=\{x\in X:x_\gamma^*(\abs{x})=0\}$ and $C_\gamma:=N_\gamma^{\rm
d}:=\{x\in X:\abs{x}\wedge \abs{y}=0\mbox{ for all }y\in N_\gamma\}$ for each
$\gamma$.
Observe that $N_\gamma$ is a band. Indeed, it is clearly an ideal. 
If a net $(x_\alpha)$ in $N_\gamma$ converges in order to some $x\in X$, then
$\abs{x_\alpha-x}\xrightarrow{o}0$ implies that
$x^*_\gamma(\abs{x})=\bigabs{x^*_\gamma(\abs{x_\alpha})-x_\gamma^*(\abs{x})}\leq
x^*_\gamma(\bigabs{\abs{x_\alpha}-\abs{x}})\leq
x^*_\gamma(\abs{x_\alpha-x})\rightarrow0$, and consequently,
$x^*_\gamma(\abs{x})=0$, i.e., $x\in N_\gamma$.
This yields the \emph{band decomposition} $X=N_\gamma\oplus C_\gamma$ by
\cite[Theorem~1.42]{ALIP:06}.
Recall from  \cite[Corollary~5.22]{AA:02} that $X$ has a weak unit $u>0$, i.e.,
any function $x\in X$ is supported in $\{\omega:u(\omega)>0\}$ off a null set.
Write $u=f_\gamma+e_\gamma$ where $f_\gamma\in N_\gamma$ and $e_\gamma\in
C_\gamma$. Since $f_\gamma\wedge e_\gamma=0$, it is easily seen that there
exists $A_\gamma\in\Sigma$ such that
$e_\gamma=u\chi_{A_\gamma}$ and $f_\gamma=u\chi_{A_\gamma^c}$. 
Each function in $N_\gamma$ is disjoint with $e_\gamma$ and is thus
supported in $A_{\gamma}^c$ off a null set; each function in $C_\gamma$ is
disjoint with $f_\gamma$ and is thus supported in $A_\gamma$ off a null set.
By countable sup property of $X$, we choose $\{\gamma_n\}_{n=1}^\infty$ such
that 
$\sup_ne_{\gamma_n}=\sup_\gamma e_\gamma$.
If $x^*_{\gamma_n}(x)=0$ for all $n\geq 1$ and some $x\in X_+$, then we have
$x\in
N_{\gamma_n}$, so that $x\wedge e_{\gamma_n}=0$, for all $n\geq 1$.
It follows that $x\wedge \sup_\gamma e_\gamma=x\wedge \sup_n
e_{\gamma_n}=\sup_n(x\wedge e_{\gamma_n})=0$, and consequently, $x\wedge
e_\gamma=0$ for any $\gamma$. This implies that $x$ is supported in $A_\gamma^c$
off a null set and hence belongs to $N_\gamma$, i.e., $x^*_\gamma(x)=0$.
\end{proof}

\begin{proof}[Proof of Theorem~\ref{arbitrage}]
It is clear that $C:=\overline{M_0-X_+}^{\sigma(X,X^\sim_n)}$ is a
$\sigma(X,X^\sim_n)$-closed cone with $-X_+ \subset C$ and $C \cap X_+=\{0\}$
because of no free lunches. Thus by Proposition~\ref{Kreps-Yan}, there exists a
strictly positive linear functional $x^* \in X^\sim_n$ such that $x^*|_C \leq
0$. This
last condition implies that $x^*|_{M_0}=0$, so that $\ker \pi=M_0 \subset
\ker(x^*|_M)$. Hence, there exists $\lambda>0$ such that $\pi=\lambda x^*|_M$.
Therefore, for each $g\in L_0(\sigma(f))\cap
X=\overline{O_{f}}^{\sigma(X,X_n^\sim)}$, it is easily seen that the price
$p:=\lambda x^*(g)$ is consistent with $(M,\pi)$.
\end{proof}

\noindent\textbf{Acknowledgements.}
The authors are grateful to the anonymous referee and editors for their many
suggestions, which greatly improved the accessibility of
the paper.
The first author also thanks the Faculty of
Science at Ryerson University for the hospitality received during his visit
there.

\end{document}